\documentclass[hidelinks,onefignum,onetabnum]{siamart251216}

\usepackage{lipsum}
\usepackage{amsfonts}
\usepackage{graphicx}
\usepackage{epstopdf}
\usepackage{algorithmic}

\usepackage{subcaption}
\usepackage{placeins}
\usepackage{verbatim}
\usepackage{color}
\usepackage{fancyhdr}

\newcommand{\calA}{\mathcal A}
\newcommand{\calS}{\mathcal S}
\newcommand{\calI}{\mathcal I}

\newcommand{\psucc}{\mathbb P_{\mathrm{succ}}}

\definecolor{ddk}{RGB}{30,144,255}

\definecolor{mjx}{RGB}{230,150,0}

\newcommand{\ket}[1]{\left| #1 \right\rangle}
\newcommand{\bra}[1]{\left\langle #1 \right|}
\newcommand{\braket}[2]{\left\langle #1 \middle| #2 \right\rangle}

\newcommand{\ti}[1]{\widetilde{#1}}

\ifpdf
  \DeclareGraphicsExtensions{.eps,.pdf,.png,.jpg}
\else
  \DeclareGraphicsExtensions{.eps}
\fi

\newsiamremark{remark}{Remark}
\newsiamremark{hypothesis}{Hypothesis}
\crefname{hypothesis}{Hypothesis}{Hypotheses}
\newsiamthm{claim}{Claim}
\newsiamremark{fact}{Fact}
\crefname{fact}{Fact}{Facts}

\headers{Random Grover Search}{D. Dong, Y. Li, and J. Ma}

\title{Random Grover Search\thanks{Submitted to the editors DATE.\funding{This work was supported in part by the National Natural Science Foundation of China under
Grant Nos. 12271109 and 12526211; by the Shanghai Pilot Program for Basic Research-Fudan
University under Grant No. 21TQ1400100 (22TQ017); by the Scientific Research Innovation Capability Support Project for Young Faculty under Grant No. SRICSPYF-ZY2025159; and by the
Xuemin Institute of Advanced Studies, Fudan University.}}}

\author{Dekuan Dong\thanks{School of Mathematical Sciences, Fudan University, Shanghai, 200433 China (\email{dkdong21@m.fudan.edu.cn}). These authors contributed equally to this work.}
\and Yingzhou Li\thanks{School of Mathematical Sciences, Fudan University, Shanghai, 200433 China, and Shanghai Key Laboratory for Contemporary Applied Mathematics, Shanghai, China (Corresponding author: \email{yingzhouli@fudan.edu.cn}).}
\and Jiaxin Ma\thanks{School of Mathematical Sciences, Fudan University, Shanghai, 200433 China (\email{majx23@m.fudan.edu.cn}). These authors contributed equally to this work.}}

\usepackage{amsopn}

\ifpdf
\hypersetup{
  pdftitle={Random Grover Search},
  pdfauthor={D. Dong, Y. Li, and J. Ma}
}
\fi

\begin{document}

\maketitle

\begin{abstract}
Grover's algorithm achieves a quadratic speedup for unstructured search given a global oracle for the target set. In many applications, however, the target set is specified as the intersection of multiple constraint sets. Constructing a global oracle for the intersection can be costly, whereas the individual constraint oracles are often much simpler to implement.

We study a randomized Grover search algorithm that directly uses these constraint oracles. At each iteration, one of the corresponding Grover operators is selected at random. For the two-operator case with uniform sampling, we prove that the success probability approaches one after
\[
\Theta \left(\frac\pi4\sqrt{\frac{N}{r}}\right)
\]
iterations, where \(r\) is the size of the intersection. Thus, the algorithm achieves the same asymptotic query complexity as standard Grover search but without requiring a global oracle. We then generalize the analysis to arbitrary sampling distributions and an arbitrary number of Grover operators through an auxiliary operator that approximates the expected Grover evolution, while retaining the same asymptotic complexity. We further show that highly biased sampling distributions can still achieve near-unit success probability, enabling cheaper Grover operators to be used more frequently. Finally, we prove asymptotic optimality and support the theoretical results with numerical simulations.
\end{abstract}



\section{Introduction}
Grover’s algorithm \cite{10.1145/237814.237866} is a fundamental quantum algorithm that achieves a quadratic speedup for unstructured search. Given oracle access to a Boolean function $f:[N]\to{0,1}$, it finds a marked item using $\mathcal{O}(\sqrt{N})$ queries, which is optimal in the black-box model \cite{doi:10.1137/S0097539796300933}. 

The standard formulation assumes a single oracle defining the solution space. In many applications, however, the solution is specified implicitly as the intersection of multiple constraints. Such settings arise naturally in constraint satisfaction, combinatorial search \cite{PhysRevA.67.022314,Chen2022,Li2024resourceefficient}, and database filtering tasks \cite{salman2012quantum}. Related work on quantum set operations has been studied in \cite{Pang2013SetOperation,elgendy2024efficient}, but existing Grover-based approaches are restricted to two-set intersections and rely on problem-specific spectral analyses that do not generalize easily. More recent structured or partial-search methods \cite{bolton2024accelerated} suggest potential improvements but typically require stronger oracle assumptions or specialized settings. Randomized or dynamic variants of Grover search have also been considered \cite{Li2024resourceefficient,chakrabarty2017dynamic,Gilliam2021GAS}, highlighting the potential of stochastic modifications.

Formally, suppose the solution set is
\[\calA = \bigcap_{i=0}^{k-1} \calA_i\]
where each $\mathcal{A}_i$ can be checked by an efficient oracle. While one may construct a global oracle for $\mathcal{A}$ and apply Grover’s algorithm directly \cite{Li2024resourceefficient}, this is often inefficient, as combining constraint oracles into a single unitary can incur substantial overhead and circuit complexity \cite{jaques2019implementinggroveroraclesquantum}. This motivates the question:
\begin{quote}
    \sloppy Can we solve the intersection search problem by directly using the simpler subproblem oracles, without explicitly constructing the global oracle?
\end{quote}
A natural idea is to associate each subset $\calA_i$ with its corresponding Grover operator $G_i$, and attempt to combine them to amplify the amplitude of the intersection $\cap_i \calA_i$. However, this approach faces an immediate difficulty: unlike the standard Grover operator, different $G_i$ generally do not commute, and their composition does not preserve the two-dimensional invariant subspace structure that underlies standard Grover algorithm. As a result, deterministic compositions of these operators are difficult to analyze and do not admit rigorous performance guarantees for coherent amplitude amplification toward the desired solution space. 

In this work, we analyze a different approach based on randomization. Instead of deterministically composing Grover operators, we consider a stochastic process in which, at each step, a Grover operator $G_i$ is selected at random and applied to the current state. This leads to a randomized sequence of unitary transformations, whose collective behavior is nontrivial but analyzable. Similar ideas of randomized unitary compositions have been successfully used in Hamiltonian simulation, where random product formulas can approximate target dynamics efficiently (e.g., qDRIFT and its higher order variants such as qSWIFT 
\cite{nakaji2023qswift,campbell2019qdrift}, random permutation of the product formula terms \cite{ChildsSu2019Randomization,Zhang2012RandomizedHamiltonian}). However, in contrast to prior work that focuses on simulation accuracy, our goal is to understand whether such randomized dynamics can serve as a search mechanism.

The main contributions of this paper are summarized as follows. 
\begin{itemize}
    \item We analyze the randomized Grover algorithm with two operators under uniform sampling. Exploiting the special structure of the expected Grover operator, we show that after 
    \[\Theta\left(\frac\pi4 \sqrt{\frac Nr}\right)\]
    iterations, the success probability reaches $1 - \mathcal O(1/\sqrt{N})$, where $r:= |\calA_0 \cap \calA_1|$. Thus, the proposed algorithm achieves the same asymptotic performance as standard Grover search without requiring a global oracle that marks elements in $\calA_0 \cap \calA_1$.
    
    \item We extend the analysis of the uniform two-operator setting to the general case of multiple Grover operators sampled according to an arbitrary probability distribution. The key ingredient is an auxiliary operator that approximates the expected Grover evolution in the relevant subspace while preserving the essential dynamics of the uniform two-operator case. Using this construction, we establish the same asymptotic performance guarantee in the general setting.

    \item We show that a high success probability can be achieved using a highly biased sampling distribution. This is particularly advantageous when the implementation costs of the Grover operators differ significantly, as it allows the less expensive operators to be sampled more frequently.

    \item We prove that the proposed algorithm is asymptotically optimal: any deterministic sequence of Grover operators that achieves a comparable success probability must use the same asymptotic number of Grover iterations.
    
    \item We provide numerical simulations that verify the theoretical results.
\end{itemize}

The remainder of this paper is organized as follows. In \Cref{sec:pre}, we briefly review the standard Grover algorithm. \Cref{sec:random_grover} developes the randomized Grover framework. Specifically, \Cref{sec:two_grover} analyzes the two-operator and uniform sampling case, \Cref{sec:over_two_grover} extends the analysis to multiple operators with non-uniform sampling distribution; \Cref{sec:biased_distribution} studies biased sampling distributions, and \Cref{sec:optimality} establishes the asymptotic optimality of the proposed algorithm. Numerical results are presented in \Cref{sec:numeric}. Finally, \Cref{sec:conclusion} concludes the paper. Additional technical details and omitted calculations are deferred to the Appendix.

\section{Preliminaries}\label{sec:pre}
\subsection{Grover's algorithm}
Grover’s algorithm \cite{10.1145/237814.237866} is a quantum search procedure for finding a marked element in an unstructured database of size $N = 2^n$, providing a quadratic speedup over classical methods. Let $\mathcal{M} \subseteq \{0,\dots,N-1\}$ denote the set of marked items with $|\mathcal{M}| = r$. The algorithm initializes the uniform superposition
\[
\ket{\psi} = \frac{1}{\sqrt{N}} \sum_{x=0}^{N-1} \ket{x}.
\]
and applies the Grover operator $G = DO$, where $O$ is the phase oracle marking elements in $\mathcal{M}$ and $D = 2\ket{\psi}\bra{\psi} - I$ is the diffusion operator.

A standard fact is that the dynamics are confined to the two-dimensional subspace spanned by the normalized marked and unmarked states
\[
\ket{\alpha} = \frac{1}{\sqrt{r}} \sum_{x \in \mathcal{M}} \ket{x}, 
\quad
\ket{\beta} = \frac{1}{\sqrt{N-r}} \sum_{x \notin \mathcal{M}} \ket{x}.
\]
In this basis, the evolution reduces to a rotation by a fixed angle determined by $\sin\theta = \sqrt{r/N}$, so that after $t$ iterations the state is
\[
\ket{\psi_t} = \sin((2t+1)\theta)\ket{\alpha} + \cos((2t+1)\theta)\ket{\beta}.
\]
As a consequence, the success probability is maximized after $\mathcal{O}(\sqrt{N/r})$ iterations.

\subsection{Notations}\label{sec:notations}
Given a search space $\calS$ of size $N=2^n$, let $\calA_{0},\calA_{1},\dots,\calA_{k-1} \subset \calS$ be $k$ distinct subsets, and let $G_{0},G_{1},\dots,G_{k-1}$ denote the corresponding Grover operators. Define 
\[\ket{\calA_i} = \frac1{\sqrt{|\calA_i|}} \sum_{x \in \mathcal A_i} \ket{x}, \quad \text{and} \quad \ket{\calA_i^c} = \frac1{\sqrt{|\calA_{i}^c|}} \sum_{x \not \in \mathcal A_i} \ket{x},\quad i = 0, 1, \dots, k-1. \]
Then each $G_i$ acts as a rotation in the two-dimensional subspace spanned by $\ket{\calA_i}$ and $\ket{\calA_i^c}$. For any integer $j \in [0, 2^{k}-1]$ with binary representation $(j_{k-1}\cdots j_1 j_0)_2$, define 
\begin{equation}\label{eq:I}
    \calI_j := \left(\bigcap_{i:j_i = 0} \calA_{i}\right) \bigcap \left(\bigcap_{i:j_i = 1} \calA_{i}^c\right).
\end{equation}
The relation between $\calA_i$ and $\calI_j$ is illustrated by the $k=3$ case in \Cref{fig:placeholder}. The corresponding normalized quantum state is
\[\ket{\calI_j} := \frac{1}{\sqrt{|\calI_j|}}\sum_{x \in \calI_j} \ket{x}.\]
By construction, the sets $\calI_i$ and $\calI_j$ are disjoint for $i\not = j$, and hence $\braket{\calI_i}{\calI_j} = 0$. The uniform superposition state $\ket{\psi}$ can be represented as
\[\ket{\psi} = \frac{1}{\sqrt N} \sum_{i=0}^{N-1} \ket{i} = \frac{1}{\sqrt N} \sum_{j=0}^{2^k-1} \sqrt{|\calI_j|} \ket{\calI_j} = \sum_{j=0}^{2^k-1} \sqrt{\frac{|\calI_j|}{N}} \ket{\calI_j} =: \sum_{j=0}^{2^k-1} \alpha_{j} \ket{\calI_{j}},\]
where $\alpha_j = \sqrt{|\calI_j| / N}$.
\begin{figure}
    \centering
    \includegraphics[width=0.8\linewidth]{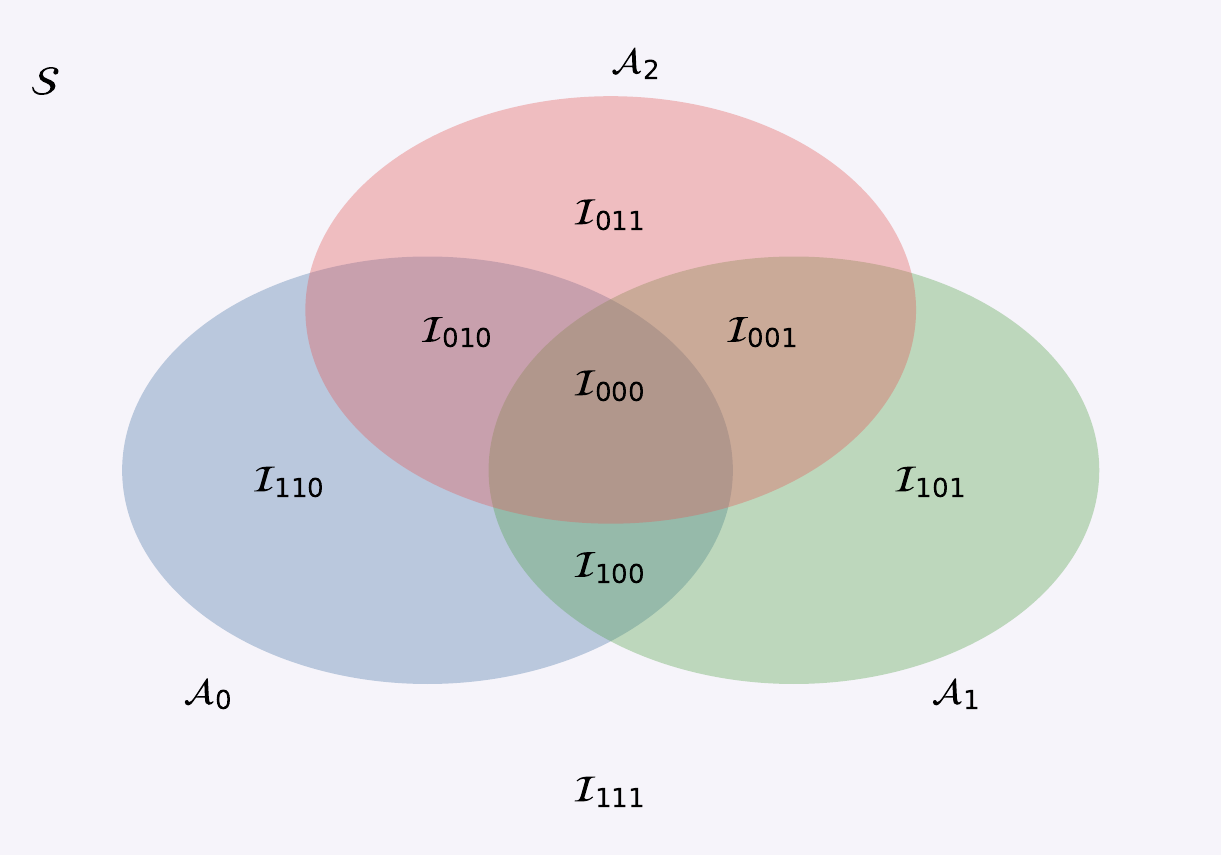}
    \caption{Partition of the search space into the regions
    $\mathcal I_j$ induced by
    $\mathcal A_0,\mathcal A_1,\mathcal A_2$.}
    \label{fig:placeholder}
\end{figure}

Note that $G_i = DO_i$, where 
\[D = 2\ket{\psi}\bra{\psi} - I, \quad \text{and} \quad O_i \ket{x} = \begin{cases} -\ket{x}, & x\in \calA_i,\\ \ket{x}, & x\in \calA_i^c.\end{cases}\]
For any operator $M \in \mathbb C^{2^n\times 2^n}$, let $[M]_\calI$ denote the restriction of $M$ to the subspace
\[\mathrm{span}\{\ket{\calI_j}: j = 0, \dots, 2^k - 1\},\]
namely 
\[[M]_\calI = \begin{bmatrix}
    \bra{\calI_0}M\ket{\calI_0} & \cdots & \bra{\calI_0}M\ket{\calI_{2^k-1}}\\
    \vdots & & \vdots \\
    \bra{\calI_{2^k-1}}M\ket{\calI_0} & \cdots & \bra{\calI_{2^k-1}}M\ket{\calI_{2^k-1}}
\end{bmatrix} \in \mathbb C^{2^k\times 2^k}.\]
Under this notation, the diffusion operator $D$ admits $[D]_\calI = 2\alpha \alpha^\top - I$, where $\alpha = [\alpha_0, \dots, \alpha_{2^k-1}]^\top$. Moreover, for every $i=0, \dots, k-1$ and every $j = 0, \dots, 2^k-1$, the set $\calI_j$ is either entirely contained in $\calA_i$ or entirely contained in its complement $\calA_i^c$. Consequently, 
\begin{equation}
    O_i \ket{\calI_j} = \begin{cases}-\ket{\calI_j}, & \calI_j \subset \calA_i,\\
\ket{\calI_j}, & \calI_j \subset \calA_i^c.\end{cases}
\end{equation}
Moreover, according to the definition of $\calI_j$, we have
\begin{equation}
    \label{eq:OiIj}
    O_i \ket{\calI_j} = (-1)^{j_i + 1} \ket{\calI_j},
\end{equation}
so $[O_i]_\calI$ is diagonal, whose $j$-th diagonal entry is $(-1)^{j_i+1}$, where $j_i$ denotes the $i$-th bit in the binary representation of $j$.

Moreover, we denote $r := |\cap_{i=0}^{k-1}\calA_i|$ and $m := |\cap_{i=0}^{k-1}\calA_i|$ for notation simplification in this paper.

\section{Random Grover Search}\label{sec:random_grover}
In this section, we sequentially apply randomly chosen operators from $\{G_0, G_1, \dots, G_{k-1}\}$ to $\ket{\psi}$, and denote by $\ket{\psi(t)}$ the state after $t$ applications. Since the subspace spanned by $\left\{\ket{\calI_j}: j = 0, \dots, 2^k -1\right\}$ is invariant under each $G_i$, there exists real coefficients $a_j (t)$ such that
\[\ket{\psi(t)} = \sum_{j=0}^{2^k-1} a_j(t) \ket{\calI_j}.\]
Note that the coefficients $a_j(t)$ are random variables, where the randomness arises from the random choice of Grover operators. For any fixed $t$, there are only finitely many possible Grover sequences, and hence each $a_j(t)$ takes values from a finite set $V_j(t)$. Measuring $\ket{\psi(t)}$ in the computational basis yields an outcome $x\in \calI_0 = \cap_i \calA_i$ with probability 
\[\begin{aligned}
    \psucc(t) &= \sum_{s\in V_0(t)} \mathbb P\left(x\in \calI_0 \big | a_0(t) = s\right) \mathbb P\left(a_0(t) = s\right) \\
    & = \sum_{s\in V_0(t)} s^2 \cdot \mathbb P\left(a_0(t) = s\right) = \mathbb E\left[a_0^2(t)\right].
\end{aligned}\]
Here, the success probability incorporates both the randomness of the choice of Grover operators and the randomness of the final measurement. The above identity shows that $\psucc(t)$ is equal to the expectation of $a_0^2(t)$, where the randomness comes solely from the choice of the Grover sequence.

\subsection{Two Grover operators}\label{sec:two_grover}
First, we consider the simplest case, where $k = 2$ and $G_0, G_1$ are selected with equal probability at each iteration. Restricting to the subspace spanned by $\left\{\ket{\calI_0}, \ket{\calI_1}, \ket{\calI_2}, \ket{\calI_3}\right\}$, the oracle operators take the form
\begin{equation}\label{eq:O01}
    [O_0]_\calI = \begin{bmatrix}-1\\ & -1 \\ && 1 \\ &&& 1 \end{bmatrix} , \quad \quad [O_1]_\calI = \begin{bmatrix}-1 \\ & 1 \\ && -1 \\ &&& 1\end{bmatrix}.
\end{equation}
The main result for this setting is stated below. 

\begin{theorem}\label{the:1}
Consider two distinct Grover operators \(G_0\) and \(G_1\), and let the initial state be the uniform superposition state $\ket{\psi}$. At each step, independently apply either \(G_0\) or \(G_1\) with equal probability. Assume $N > (m+r)^2 / r$, then after
\[
T = \Theta\left(\frac{\pi}{4}\sqrt{\frac{N}{r}}\right)
\]
steps, measuring the resulting quantum state yields an element in \(\calA_0\cap\calA_1\) with probability
\[
\mathbb P_{\mathrm{succ}}(T)
=
\mathbb E\left[a_0^2(T)\right]
\geq
1-\left(
\frac{2(m-r)}{\sqrt{Nr}}
+
\frac{4r}{N}
\right).
\]
\end{theorem}
\begin{proof} 
By the Cauchy inequality, we have
\begin{equation}
    \mathbb{E}\left[a_0^2(t)\right] \geq \left(\mathbb{E}\left[a_0(t)\right]\right)^2.
\end{equation}
Since the random choices are independent, it follows that
\begin{equation}
    \begin{aligned}
        \mathbb E\left[a_{0}(t)\right] &=\mathbb{E}\!\left[\bra{\calI_{0}}G_{i_t}G_{i_{t-1}} \cdots G_{i_1}\ket{\psi}\right] = \bra{\calI_{0}}\mathbb{E}\left[G_{i_t}G_{i_{t-1}} \cdots G_{i_1}\right]\ket{\psi} \\
        &= \bra{\calI_{0}}\mathbb{E}\left[G_{i_t}\right]\mathbb{E}\left[G_{i_{t-1}}\right]\cdots\mathbb{E}\left[G_{i_1}\right]\ket{\psi} = \bra{\calI_{0}}\mathbb{E}\left[G_{i_{1}}\right]^{t}\ket{\psi} \\
        &= \bra{\calI_{0}}\left(\frac{G_{0}+G_{1}}{2}\right)^{t}\ket{\psi}.
    \end{aligned}
\end{equation}
For convenience, we define 
\begin{equation}\label{eq:G}
    G:= (G_0 + G_1) / 2.
\end{equation} 
According to \cref{eq:O01}, in the subspace spanned by $\{\ket{\calI_0}, \ket{\calI_1}, \ket{\calI_2}, \ket{\calI_3}\}$, the operator $G$ takes the form 
\[
   [G]_\calI = (2\alpha \alpha^\top - I)\left(\frac{[O_0]_\calI + [O_1]_\calI}{2}\right)= \begin{bmatrix}
    1-2\alpha_{0}^{2} & 0 & 0 & 2\alpha_{0}\alpha_{3}  \\
    -2\alpha_{0}\alpha_{1} & 0 & 0 & 2\alpha_{1}\alpha_{3}  \\
    -2\alpha_{0}\alpha_{2} & 0 & 0 & 2\alpha_{2}\alpha_{3}  \\
    -2\alpha_{0}\alpha_{3} & 0 & 0 & 2\alpha_{3}^{2}-1 
    \end{bmatrix}.
\]
Next, suppose
\[G^t \ket{\psi} = \sum_{j=0}^3 \widehat a_j(t) \ket{\calI_j}.\]
Then the coefficients satisfy the recursion

\begin{equation} \label{equ: M_evolution}
\begin{bmatrix}
\widehat a_{0}(t) \\\widehat a_{3}(t)
\end{bmatrix}=M \begin{bmatrix}
\widehat a_{0}(t-1) \\
\widehat a_{3}(t-1)
\end{bmatrix},\quad \text{where} \quad 
M = \begin{bmatrix}
1 - 2 \alpha_0^2 & 2\alpha_0 \alpha_3 \\
-2 \alpha_0 \alpha_3 & 2\alpha_3^2 -1
\end{bmatrix}, \; \begin{bmatrix}
    \widehat a_0(0)\\ \widehat a_3(0)
\end{bmatrix} = \begin{bmatrix} \alpha_0\\ \alpha_3\end{bmatrix}.
\end{equation}
Here, $\alpha_0 = \sqrt{r/N}$ and $\alpha_3 = \sqrt{1-m/N}$.
Applying $M^t$ to the initial vector $[\alpha_0, \alpha_3]^\top$ and extracting the first component yields
\begin{equation}\label{eq:hat_a_0_t}
    \widehat a_0(t)
    = \frac{\alpha_0(1-\alpha_0^{2}+\alpha_3^2 + \Delta)(\alpha_3^2 - \alpha_0^2 + \Delta)^t 
    + \alpha_0(\alpha_0^{2}-\alpha_3^2 -1 + \Delta)(\alpha_3^2 - \alpha_0^2 - \Delta)^t}{2\Delta},
\end{equation}
where $\Delta = \sqrt{(1- \alpha_0^2 - \alpha_3^2)^2 - 4\alpha_0^2 \alpha_3^2}$. The derivation of the above expression is deferred to \Cref{app:hat_a_0_t}. Since $N > (m+r)^2 / r$, we have  
\[(\alpha_0 + \alpha_3)^2 = \frac{N - m + r + 2\sqrt{r(N-m)}}{N} > \frac{N - m + r + 2\sqrt{m^2 + mr + r^2}}{N} > 1, \]
and thus $(1- \alpha_0^2 - \alpha_3^2)^2 - 4\alpha_0^2 \alpha_3^2 < 0$. Therefore, we write $\Delta = \imath \Gamma$, where $\Gamma = \sqrt{4\alpha_0^2 \alpha_3^2 - (1- \alpha_0^2 - \alpha_3^2)^2}$, and the expression for $\widehat a_0(t)$ can be rewritten as
\[\begin{aligned}
    \widehat a_0(t) &= \frac{2\alpha_0\alpha_3 e^{\imath \beta} \rho^t e^{\imath \theta t} - 2\alpha_0\alpha_3 e^{-\imath \beta}\rho^t e^{-\imath \theta t} }{2\imath \Gamma} =  \frac{2\alpha_0\alpha_3\rho^t\sin(t\theta+\beta)}{\Gamma},
\end{aligned}\]
where 
\[\rho = \sqrt{2(\alpha_0^2 + \alpha_3^2)-1}, \quad \sin \theta = \frac{\Gamma}{\rho}, \quad \sin \beta = \frac{\Gamma}{2\alpha_3}, \quad \theta, \beta \in \left[0, \frac \pi 2\right].\]

Let 
\[T = \left\lfloor \frac{\pi}{2\theta} - \frac{\beta}{\theta} \right\rfloor > \frac{\pi}{2\theta} - \frac{\beta}{\theta} - 1.\]
Then we have $T \theta + \beta \in (\pi/2 - \theta, \pi/2]$, and thus $\sin(T \theta + \beta) \ge \sin (\pi / 2 - \theta) = \cos\theta$. Moreover, noting that $\Gamma \leq 2\alpha_0\alpha_3$, we obtain
\begin{equation}
\widehat a_{0}(T_{0}) \geq \rho^{T_{0}}\cos\theta
=  \cos\theta \left(1 - 2(1 - \alpha_0^2 - \alpha_3^2)\right)^{\frac{T_{0}}{2}} \geq 
\cos\theta \left(1 - (1 - \alpha_0^2 - \alpha_3^2)T_{0}\right),
\label{equ:T0}
\end{equation}
where the last inequality follows from Bernoulli's inequality.

Consequently,
\begin{equation}\label{eq:exp_succ}
\begin{aligned}
    0 &\le 1 - \mathbb{E}\left[a_0^2(t)\right] \le 1 - \left(\mathbb E\left[a_0(t)\right]\right)^2 = 1 - \widehat a_0^2(t)\\
    &\le 1 - \cos^2\theta \left(1 - (1 -\alpha_0^2 - \alpha_3^2)T\right)^2\\
    &\le \sin^{2}\theta + 2 \cos^{2}\theta (1 - \alpha_0^2 - \alpha_3^2)T_{0} \\
    &\le \sin^2 \theta + 2\cdot \frac{m-r}{N} \cdot \frac{\pi}{2\theta} \le \sin^2 \theta + \frac{m-r}{N} \cdot \frac{\pi}{\sin \theta}.
\end{aligned}
\end{equation}
Here, 
\[\sin \theta = \frac{\Gamma}{\rho} = \sqrt{\frac{\frac{4r(N-m) - (m-r)^2}{N^2}}{1- \frac{2(m-r)}{N}}}
 = 2\sqrt{\frac rN} \sqrt{\frac{1-\left(\frac{m+r}{2}\right)^2 \frac{1}{N r}}{1 - 2\frac{m-r}{N}}} \le 2\sqrt{\frac rN}.\]
Moreover, since $N > (m+r)^2 / r$, it follows that
\[ \frac{1}{\sin \theta} \le \frac12 \sqrt{\frac Nr} \sqrt{\frac 43} = \frac1{\sqrt 3} \sqrt{\frac Nr}. \]
Substituting these bound into \cref{eq:exp_succ}, we obtain 
\[0 \le 1 - \mathbb E\left[a_0^2(t)\right] \le \frac{2(m-r)}{\sqrt{Nr}} + \frac{4 r}{N}.\]
For the oracle complexity, consider the function
\[
f(r) := \frac{\sin2\beta}{\sin\theta} = \frac{(1+\alpha_3^2-\alpha_0^2)\sqrt{2(\alpha_0^2 + \alpha_3^2)-1}}{2\alpha_3^2}
= \frac{(2-\frac{m+r}{N})\sqrt{1-\frac{2(m-r)}{N}}}{2-\frac{2m}{N}},
\]
for $r \in (0,m]$. Using the assumption $N > (m+r)^2 / r$, it follows that
\[
f'(r) = \frac{N+m-3r}{2N(N-m)\sqrt{1-\frac{2(m-r)}{N}}} \geq \frac{3m - 2r}{2N^2} > 0, \qquad \forall r \in (0,m].
\]
Hence $f(r) \leq f(m) = 1$. Then we get $2\beta \leq \theta$, and thus $0< \beta / \theta \leq 1/2$. Then, we have
\[
\begin{aligned}
    T_{0} \geq \frac{\pi}{2\theta} - 2 \geq \frac{\pi}{2(\sin\theta + \frac{\theta^3}{6})} -2 \geq \frac{\pi}{2(\sin\theta + \frac{1}{6}(\frac{\pi}{2}\sin\theta)^3)} -2 \geq \frac{\pi}{4}\sqrt{\frac{N}{r}}\frac{N}{N+4r} -2.
\end{aligned}
\]
On the other hand, we have
\[
\begin{aligned}
    T \leq \frac{\pi}{2\theta} \leq \frac{\pi}{2\sin\theta} 
    \leq \frac{\pi}{4}\sqrt{\frac{N}{r}}\frac{1}{\sqrt{1-\frac{(m+r)^2}{4Nr}}}.
\end{aligned}
\]
Finally, we conclude that $T = \Theta(\pi/4\sqrt{N/r})$.
\end{proof}

\begin{remark}
The operator \(M\) can be written as
\[
M =
\left(
I -
2
\begin{bmatrix}
\alpha_0\\
\alpha_3
\end{bmatrix}
\begin{bmatrix}
\alpha_0 & \alpha_3
\end{bmatrix}
\right)
\begin{bmatrix}
1 & 0\\
0 & -1
\end{bmatrix},
\]
which closely resembles the Grover operator. However, since $\alpha_0^2+\alpha_3^2<1$, this operator is not a perfect rotation and instead leads to leakage of amplitude outside the two-dimensional subspace during the evolution. \Cref{fig:fig1_3img} shows the trajectory of \([\widehat a_0(t),\widehat a_3(t)]^\top\) in the two-dimensional plane under different parameters.
\begin{figure}[!ht]
    \centering
    \includegraphics[width=0.9\textwidth]{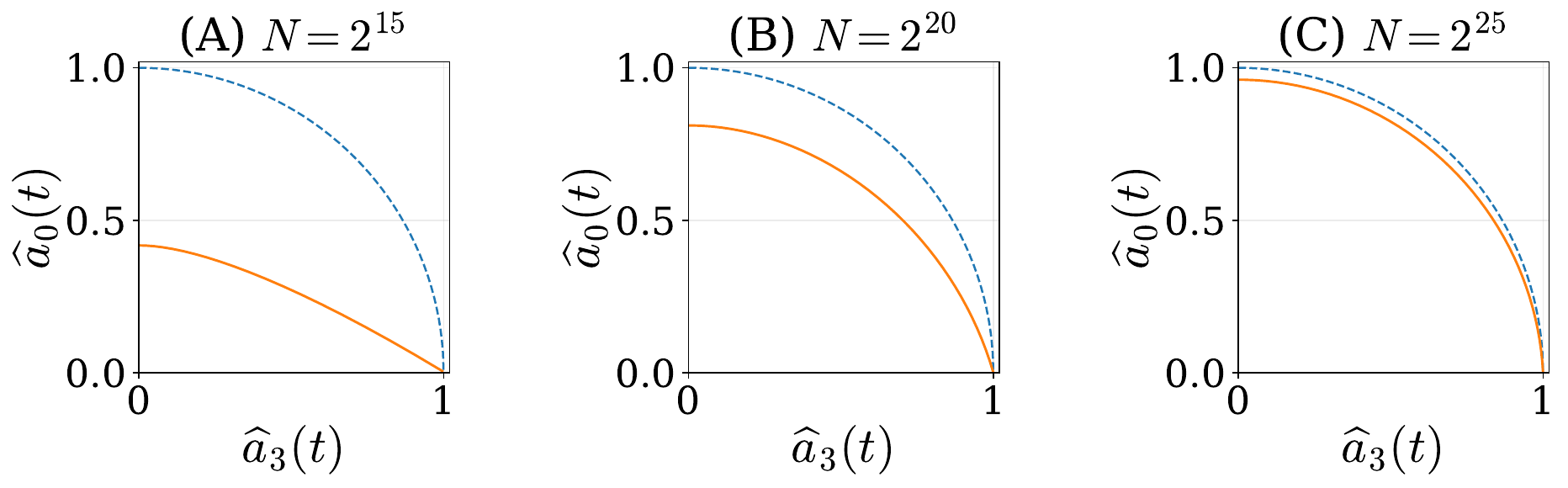}
    \caption{Projected trajectories onto the \((\widehat a_3(t),\widehat a_0(t))\)-plane under the averaged evolution matrix \(M\) for different values of \(N\), with fixed \(r=1\), \(|\mathcal A_0|=200\), and \(|\mathcal A_1|=100\). 
    The dashed curves represent quarter circles of radius \(1\), while the solid curves denote the projected trajectories.}
    \label{fig:fig1_3img}
\end{figure}
\end{remark}

\subsection{More than two Grover operators}\label{sec:over_two_grover}
In this section, we extend the analysis to the case of more than two Grover operators. Let $\{G_i:i = 0, \dots, k-1\}$ be selected independently at each step with probability $\{p_i\}_{i=0}^{k-1}$. As discussed eariler, 
\[\mathbb E[a_0(t)] = \bra{\calI_0} \left(\sum_{i=0}^{k-1} p_i G_i\right)^t \ket{\psi}.\]
Using the notations introduced in \Cref{sec:notations}, we define 
\[G := \sum_{i=0}^{k-1} p_iG_i = D \sum_{i=0}^{k-1} p_iO_i.\]
Restricted to the subspace 
\[\mathrm{span} \left\{\ket{\calI_j}: j = 0, \dots, 2^k-1\right\},\]
we have $[D]_\calI = 2\alpha\alpha^\top - I$ and 
\[\left[\sum_{i=0}^{k-1} p_iO_i\right]_\calI = \sum_{i=0}^{k-1} p_i \left[O_i\right]_\calI.\]
According to \cref{eq:OiIj}, each $[O_i]_\calI$ is diagonal, with $j$-th diagonal entry equal to $(-1)^{1 + j_i}$. Consequently, $\sum_{i=0}^{k-1} p_i \left[O_i\right]_\calI$ is also diagonal, whose $j$-th diagonal entry is 
\begin{equation}\label{eq:O_jj}
    \sum_{i=0}^{k-1} p_i (-1)^{1 + j_i} = \sum_{i:j_i=1} p_i - \sum_{i:j_i=0} p_i = 1 - 2\sum_{i:j_i=0} p_i,
\end{equation}
where we use the identity $\sum_{i=0}^{k-1} p_i = 1$.

Unlike the two-operator case with uniform sampling, the corresponding average evolution matrix \(G\) is no longer sparse in the relevant subspace, which makes its spectral analysis and the characterization of its powers considerably more challenging. To address this difficulty, we construct a sparse approximation $\widetilde{G}$ by discarding entries whose contributions are asymptotically negligible. The following lemma provides the construction of $\widetilde{G}$.
\begin{lemma}\label{lemma:1}
Consider $k\ge 2$ distinct Grover operators $G_0, G_1, \dots, G_{k-1}$. Let
\[G = \sum_{i=0}^{k-1} p_i G_i, \qquad \sum_{i=0}^{k-1}p_i = 1, \qquad p_i > 0,\]
define $\ti{\alpha} := [\alpha_0, 0, \dots, 0, \alpha_{2^k-1}]^\top$, and let $\ti{G}$ be an operator such that the subspace
\[\mathrm{span}\left\{\ket{\calI_j}: j = 0, \dots, 2^k - 1\right\}\]
is invariant under $\ti{G}$ and 
\[[\ti{G}]_\calI := \left(2 \alpha \ti{\alpha}^\top - I\right) \sum_{i=0}^{k-1} p_i [O_i]_\calI. \]
Then, for all $n\in \mathbb N$, 
\begin{equation}\label{equ:1}
    \left\|\left[G^n\right]_\calI - [\tilde{G}^n]_\calI\right\|_2 \leq \frac 1p \sqrt{1 - \alpha_0^2 - \alpha_{2^{k} - 1}^2} + \frac{2n}p (1 - \alpha_0^2 - \alpha_{2^{k} - 1}^2), 
\end{equation}
where 
\[p:= \min_{0\le i \le k-1} p_i.\]
\end{lemma}
\begin{proof}
First, observe that
\[ [G]_\calI - [\ti{G}]_\calI = 2\alpha (\alpha - \ti{\alpha})^\top \sum_{i=0}^{k-1}p_i [O_i]_\calI := 2\alpha\beta^{\top}.\]
Note that $\sum_{i=0}^{k-1}p_i [O_i]_\calI$ is diagonal and its $j$-th diagonal entry is given by \cref{eq:O_jj}, we have 
\[ \beta_j = \begin{cases}
    0, & j = 0, 2^k-1,\\
    \alpha_j \left(1 - 2\sum_{i:j_i = 0} p_i\right), & j = 1, \dots, 2^k - 2.
\end{cases}\]
Since $\|[G]_\calI\|_2 \le \|G\|_2 = 1$, we have
\[\left\|[G]_\calI^{n} - [\ti{G}]_\calI^{n}\right\|_2 = \left\|\sum_{l=0}^{n-1} [G]_\calI^{n-1-l}([G]_\calI-[\ti{G}]_\calI)[\ti{G}]_\calI^{l}\right\|_2
\leq \sum_{l=0}^{n-1} \left\|([G]_\calI-[\ti{G}]_\calI)[\ti{G}]_\calI^{l}\right\|_2.
\]
Moreover,
\[\left\|([G]_\calI-[\ti{G}]_\calI)[\ti{G}]_\calI^{l}\right\|_2 = 2\left\|\alpha\beta^{\top}[\ti{G}]_\calI^{l}\right\|_2 = 2\|\alpha\|_2 \left\|\left([\ti{G}]_\calI^{\top}\right)^l\beta\right\|_2 = 2\left\|\left([\ti{G}]_\calI^{\top}\right)^l\beta\right\|_2.
\]

Denote
\[
\left([\ti{G}]_\calI^{\top}\right)^{l}\beta
=
\left[b_{0}(l), b_{1}(l), \dots, b_{2^{k}-1}(l)\right]^{\top}.
\]
By writing out the expression of $[\ti{G}]_\calI$ explicitly, we obtain the following recurrence relation:
\begin{equation}\label{eq:rec_b}
    \small 
    \begin{bmatrix}
        b_0(l)\\
        b_1(l)\\
        \vdots\\
        b_{2^k-2}(l)\\
        b_{2^k-1}(l)
        \end{bmatrix}
        =
        \begin{bmatrix}
        1-2\alpha_0^2 & -2\alpha_0\alpha_1 & \cdots & -2\alpha_0\alpha_{2^k-2} & -2\alpha_0\alpha_{2^k-1}\\
        0 & \gamma_1 &  & 0 & 0\\
        \vdots &  & \ddots &  & \vdots\\
        0 & 0 &  & \gamma_{2^k-2} & 0\\
        2\alpha_{2^k-1}\alpha_0 & 2\alpha_{2^k-1}\alpha_1 & \cdots & 2\alpha_{2^k-1}\alpha_{2^k-2} & 2\alpha_{2^k-1}^2-1
        \end{bmatrix}
        \begin{bmatrix}
        b_0(l-1)\\
        b_1(l-1)\\
        \vdots\\
        b_{2^k-2}(l-1)\\
        b_{2^k-1}(l-1)
    \end{bmatrix},
\end{equation}
where
\[\gamma_j= 2 \sum_{i:j_i = 0} p_i - 1,
\qquad
j = 1, \dots, 2^k-2.\]
Let $\gamma = \max_{2 \leq j \leq 2^{k}-2} |\gamma_{j}|$. It can be shown that 
\[2p-1 \le \gamma_j \le 1 - 2p, \quad \forall j = 1, \dots, 2^k - 2,\]
which implies $\gamma \le 1 - 2p$. 

From the recurrence \cref{eq:rec_b}, the first and last components satisfy
\[
\begin{bmatrix}
b_{0}(l) \\
b_{2^{k}-1}(l)
\end{bmatrix}
=
M
\begin{bmatrix}
b_{0}(l-1) \\
b_{2^{k}-1}(l-1)
\end{bmatrix}
+
v(l-1),
\]
where
\[
M = \begin{bmatrix}
1 - 2 \alpha_0^2 & - 2\alpha_0 \alpha_{2^{k}-1} \\
2 \alpha_0 \alpha_{2^{k}-1} & 2\alpha_{2^{k}-1}^2 -1
\end{bmatrix},
\qquad
v(l-1) = \begin{bmatrix}-2\alpha_0\left(\sum_{j=1}^{2^{k}-2}\alpha_j b_{j}(l-1)\right) \\
2\alpha_{2^{k}-1}\left(\sum_{j=1}^{2^{k}-2}\alpha_j b_{j}(l-1)\right)
\end{bmatrix}.
\]
For the remaining components, we have
\[
b_{j}(l) = \gamma_{j} b_{j}(l-1)= \dots = \gamma_{j}^{l}b_{j}(0) = \gamma_{j}^{l+1} \alpha_i,\quad \forall j = 1, \dots, 2^k-2.
\]

For the term $v(l-1)$, we have
\begin{equation}
\|v(l-1)\|_{2} = 2\sqrt{\alpha_0^2 + \alpha_{2^k-1}^2} \left |\sum_{j=1}^{2^{k}-2}\alpha_j b_{j}(l-1) \right | 
\leq 2 \left |\sum_{j=1}^{2^{k}-2}\alpha_j^2  \gamma_j^l \right |
\leq 2 \gamma^{l}(1-\alpha_0^2 - \alpha_{2^k-1}^2).
\end{equation}
Since $\|M\|_{2} = 1$ and $b_{0}(0) = b_{2^{k}-1}(0) = 0$, it follows that
\begin{equation}
\begin{aligned}
\sqrt{b_{0}(l)^2 + b_{2^{k}-1}(l)^2}
&\leq \sqrt{b_{0}(l-1)^2 + b_{2^{k}-1}(l-1)^2} + 2\gamma^{l} (1-\alpha_0^2 - \alpha_{2^k-1}^2) \\
&\leq 2\left(\sum_{q=1}^{l}\gamma^{q}\right)(1-\alpha_0^2 - \alpha_{2^k-1}^2) \\
&\leq \frac{2\gamma}{1-\gamma} (1-\alpha_0^2 - \alpha_{2^k-1}^2)\\
& \le \frac1p (1-\alpha_0^2 - \alpha_{2^k-1}^2),
\end{aligned}
\end{equation}
where we used $\gamma \le 1 - 2p$ in the last inequality. Moreover, a direct computation gives
\[
\sqrt{\sum_{j=1}^{2^{k}-2}b_{j}(l)^{2}} \leq \gamma^{l+1}\sqrt{1-\alpha_0^2 - \alpha_{2^k-1}^2}.
\]
Combining the above bounds and using the triangle inequality, we obtain
\[\left\|\left([\ti{G}]_\calI^{\top}\right)^{l}\beta\right\|_{2} \leq \frac1p (1-\alpha_0^2 - \alpha_{2^k-1}^2) + \gamma^{l+1}\sqrt{1-\alpha_0^2 - \alpha_{2^k-1}^2}.
\]
Therefore,
\begin{equation}
\begin{aligned}
\left\|[G]_\calI^{n} - [\ti{G}]_\calI^{n}\right\|_2
&\leq \sum_{l=0}^{n-1} \left\|([G]_\calI-[\ti{G}]_\calI)[\ti{G}]_\calI^{l}\right\|_2 \\
&\leq \frac{2n}p(1-\alpha_0^2 - \alpha_{2^k-1}^2) + 2\sum_{l=0}^{n-1}\gamma^{l+1}\sqrt{1-\alpha_0^2 - \alpha_{2^k-1}^2} \\
&\leq \frac{2n}{p}(1-\alpha_0^2 - \alpha_{2^k-1}^2) + \frac 1p\sqrt{1-\alpha_0^2 - \alpha_{2^k-1}^2}.
\end{aligned}
\end{equation}
Finally, note that the subspace 
\[\mathrm{span}\left\{\ket{\calI_j}: j = 0, \dots, 2^k-1\right\}\]
is invariant under both $G$ and $\ti{G}$, we have $[G^n]_\calI = [G]_\calI^n$ and $[\ti{G}^n]_\calI = [\ti{G}]_\calI^n$, and thus 
\[\left\|\left[G^n\right]_\calI - [\tilde{G}^n]_\calI\right\|_2 = \left\|[G]_\calI^{n} - [\ti{G}]_\calI^{n}\right\|_2 \leq \frac 1p \sqrt{1 - \alpha_0^2 - \alpha_{2^{k} - 1}^2} + \frac{2n}p (1 - \alpha_0^2 - \alpha_{2^{k} - 1}^2).\]
\end{proof}
\begin{remark}
Since \(p \leq 1/k\), the right-hand side of \cref{equ:1} attains its minimum at \(p = 1/k\), i.e., when all Grover operators are sampled with equal probability.
\end{remark}
Using \Cref{lemma:1}, we obtain the main result for the randomized Grover search algorithm with arbitrary number of operators and non-uniform sampling distribution.
\begin{theorem}\label{the:2}
Consider $k\ge 2$ distinct Grover operators $G_0, G_1, \dots, G_{k-1}$, and let the initial state be the uniform superposition state $\ket{\psi}$. At each step, independently apply $G_i$ with probability $p_i>0$. Assume $N > (m+r)^2/r$, then after
\[ T = \Theta\left(\frac{\pi}{4} \sqrt{\frac{N}{r}}\right) \]
steps, measuring the resulting quantum state yields an element in $\cap_{i=0}^{k-1} \calA_i$ with probability
\[\psucc(T) = \mathbb E\left[a_0^2(T)\right] \ge 1 - \frac{4r}{N} - \left(2+\frac{4}{p}\right) \frac{m-r}{\sqrt{rN}} - \frac{2}{p}\sqrt{\frac{m-r}{N}},\]
where $p = \min_{0\le i\le k-1} p_i$.
\end{theorem}
\begin{proof}
Let $\ti{G}$ be the operator constructed in \Cref{lemma:1}. Observe that the evolution induced by $\ti{G}$ on the invariant subspace $\mathrm{span}\{\ket{\calI_{0}}, \ket{\calI_{2^k-1}}\}$ is identical to the evolution on the invariant subspace $\mathrm{span}\{\ket{\calI_{0}}, \ket{\calI_{3}}\}$ considered in \Cref{the:1}. Therefore, we choose
\[ T_{0} = \left\lfloor \frac{\pi}{2\theta} - \frac{\beta}{\theta} \right\rfloor = \Theta\left(\frac{\pi}{4} \sqrt{\frac{N}{r}}\right)\]
where $\theta$ and $\beta$ are the same as in \Cref{the:1}.

Using \cref{equ:T0}, together with the estimate in \Cref{lemma:1} and the triangle inequality, we obtain
\[\begin{aligned}
    & \; \mathbb E[a_0(T)]\\
    = & \; \mathbb E \left[\bra{\calI_0}G^{T_{0}}\ket{\psi}\right] \geq \mathbb E\left[\bra{\calI_0}\ti{G}^{T_{0}}\ket{\psi}\right] - \left\|[G^{T_{0}}]_\calI - [\ti{G}^{T_{0}}]_\calI\right\|_2 \\
    \geq & \; \cos\theta(1 - (1 - \alpha_0^2 - \alpha_{2^k-1}^2)T_{0}) - \frac{1}{p}\sqrt{1 - \alpha_0^2 - \alpha_{2^{k}-1}^2} - \frac{2T_{0}}{p}(1 - \alpha_0^2 - \alpha_{2^{k}-1}^2).
\end{aligned}\]
Hence, the success probability satisfies
\[\begin{aligned}
    & \psucc(T) = \mathbb E\left[a_0^2(T)\right] \ge \left(\mathbb E\left[a_0(T)\right]\right)^2\\
    \ge \; &\cos^2 \theta \left(1 - \left(1 - \alpha_0^2 - \alpha_{2^k-1}^2\right)T\right)^2 - \frac{2}{p}\sqrt{1 - \alpha_0^2 - \alpha_{2^{k}-1}^2} - \frac{4T_{0}}{p}(1 - \alpha_0^2 - \alpha_{2^{k}-1}^2)\\
    \ge \; &  1 - \sin^2\theta - 2(1 - \alpha_0^2 - \alpha_{2^{k}-1}^2) T - \frac{2}{p}\sqrt{1 - \alpha_0^2 - \alpha_{2^{k}-1}^2} - \frac{4T_{0}}{p}(1 - \alpha_0^2 - \alpha_{2^{k}-1}^2)\\
    \ge \; & 1 - \frac{4r}{N} - \left(2+\frac{4}{p}\right) \frac{m-r}{\sqrt{rN}} - \frac{2}{p}\sqrt{\frac{m-r}{N}}, 
\end{aligned}\]
where we use the bound $\sin\theta \le 2\sqrt{\frac rN}$ and $T \le \sqrt{\frac Nr}$ from \Cref{the:1}.
\end{proof}

According to \Cref{the:2}, the randomized Grover search algorithm amplify the success probability approaching $1$ if $1/p = o(\sqrt{N})$ under mild assumptions between $m$, $r$, and $N$. In contrast, a deterministic sequence of Grover operators does not necessarily amplify the success probability to approach $1$. To illustrate this, consider the sequence
\begin{equation}\label{eq:determined_seq}
    G_0, G_1, G_0, G_0, G_1, G_0, G_0, \dots
\end{equation}
Let $r_0 := |\mathcal{A}_0|$, $r_1 := |\mathcal{A}_1|$, and $r:=|\mathcal{A}_0 \cap \mathcal{A}_1|$, and assume $N > \max\{4r_0, 4r_1\}$. The operator $G_0 G_0 G_1$ appears periodically in this sequence and can be rewritten as  
\[\begin{aligned}
    G_0 G_0 G_1 &= G_0 (2\ket{\psi}\bra{\psi} - I) O_0 (2\ket{\psi}\bra{\psi} - I) O_1\\
    & = G_0  (2\ket{\psi}\bra{\psi} - I) G_0^{\dagger} O_1 \\
    &= \left(2 G_0 \ket{\psi}\bra{\psi} G_0^{\dagger} - I \right) O_1.
\end{aligned}   
\]
Operators of this form were studied in~\cite{Grover1998}, where the Walsh--Hadamard transform in the standard Grover algorithm is replaced by a general unitary operator. Following the same line of argument, we identify a two-dimensional invariant subspace of $G_0 G_0 G_1$, spanned by
\[
\ket{u} = \frac{1}{\sqrt{|\calA_1|}} \sum_{y \in \mathcal{A}_1} \ket{y}, 
\quad 
\ket{v} = G_0 \ket{\psi}.
\]
Restricted $G_0G_0G_1$ to the subspace $\mathrm{span}\left\{\ket{u}, \ket{v}\right\}$, and expressing it in the basis $\left\{\ket{u}, \ket{v}\right\}$, we obtain the matrix representation
\begin{equation}\label{eq:W}
    W := \begin{bmatrix}
        1 & 2 \braket{u}{v} \\
        -2 \braket{v}{u} & 1 - 4 |\braket{u}{v}|^2
    \end{bmatrix}.
\end{equation}
Since $\operatorname{Tr}(W) = 2 - 4 |\braket{u}{v}|^2$ and $\det(W) = 1$, the powers of $W$ satisfy the recurrence
\[W^n 
= (2 - 4 |\braket{u}{v}|^2)W^{n-1} 
- W^{n-2}.\]
Solving this recurrence yields
\[W^n = \frac{1}{\sin\theta}
\begin{bmatrix}
\sin n\theta - \sin(n-1)\theta & 2\braket{u}{v}\sin n\theta \\
-2\braket{v}{u}\sin n\theta & (1-4|\braket{u}{v}|^2)\sin n\theta - \sin(n-1)\theta
\end{bmatrix},\]
where $\sin\frac{\theta}{2} := \braket{u}{v}$. The derivation of the expression is deferred to \Cref{app:w}.

After one step of the Grover sequence, the quantum state becomes $G_0 \ket{\psi}$, which lies in the invariant subspace spanned by $\ket{u}$ and $\ket{v}$. After $n$ applications of $G_0G_0G_1$, the success probability satisfies
\begin{equation}   ~\label{equ:p_succ_contrast}
\begin{aligned}
    \psucc^{(n)} = &\left| \bra{\mathcal{I}_0} (G_0 G_0 G_1)^n G_0 \ket{\psi} \right|^2 \\
    =& \left( \frac{2\sin\frac{\theta}{2}\sin n\theta}{\sin\theta}\braket{\calI_0}{u}
    +
    \frac{(1 - 4\sin^2\frac{\theta}{2})\sin n\theta - \sin(n-1)\theta}{\sin\theta} \braket{\calI_0}{v}
    \right)^2 \\
    =& \left(\frac{2\sin\frac{\theta}{2}\sin n\theta}{\sin\theta}\sqrt{\frac{r}{r_1}} 
    + 
    \left(3 - \frac{4r_0}{N}\right)\sqrt{\frac{r}{N}}
    \cdot
    \frac{(1 - 4\sin^2\frac{\theta}{2})\sin n\theta - \sin(n-1)\theta}{\sin\theta}
    \right)^2 \\
    =&
    \frac{4\sin^{2}\frac{\theta}{2}}{\sin^{2}\theta}
    \left(
    \sin n\theta \sqrt{\frac{r}{r_1}} 
    + 
    2\left(3 - \frac{4r_0}{N}\right)\sqrt{\frac{r}{N}}
    \left(
    \cos\left(n-\tfrac{1}{2}\right)\theta 
    - 
    2\sin\frac{\theta}{2}\sin n\theta
    \right)
    \right)^2\\
    \le & \frac{4\sin^{2}\frac{\theta}{2}}{\sin^{2}\theta}
    \left(
    \sqrt{\frac{r}{b}} + 18\sqrt{\frac{r}{N}}
    \right)^2.
\end{aligned}
\end{equation}
This upper bound can be significantly smaller than $1$ when $r \ll r_1 < N$. Moreover, the success probabilities associated with the sequences $G_1(G_0 G_0 G_1)^n G_0$ and $G_0G_1(G_0 G_0 G_1)^n G_0$ can also be bounded far below $1$. The corresponding analysis is presented in \Cref{app:deterministic-prefix}. Therefore, the deterministic sequence \cref{eq:determined_seq} provides an example in which the success probability does not necessarily approach $1$. Numerical experiments are conducted in \Cref{sec:comp_deter_rand} to verify the above analysis. 

\subsection{Biased sampling distribution}\label{sec:biased_distribution}
We now consider the scenario where $G_0$ is substantially cheaper to implement than the remaining Grover operators, and the goal is to minimize the number of queries to $G_i, i = 1, \dots, k-1$. A related problem was studied in \cite{khadiev2023intersection}, where the authors considered the two-operator setting and proposed a deterministic construction under the assumptions that $\calA_1 \subset \calA_0$ and $|\calA_1| \ll |\calA_0|$. Our framework requires neither of these restrictions. 

By adopting a suitably biased sampling distribution, we can significantly reduce the expected number of applications of the remaining operators, while applying $G_0$ in the majority of iterations. The following corollary establishes the corresponding performance guarantee.
\begin{corollary}\label{lemma: delta_k}
    Let $\delta \in (0, 1/k)$, and suppose that
    \begin{equation}\label{eq:cond_p1}
     p := \frac{1}{\delta - \frac{4r}{N} - \frac{2(m-r)}{\sqrt{rN}}} \left(\frac{4(m-r)}{\sqrt{rN}} + 2\sqrt{\frac{m-r}{N}}\right) < \frac1k.
    \end{equation}
    At each step, independently apply an operator chosen from $\{G_0, G_1, \dots, G_{k-1}\}$ according to the probability distribution
    \[
    p_0 = 1 - (k-1) p, \qquad p_1 = p_2 = \cdots = p_{k-1} = p.
    \] 
    Then, after
    \[T = \Theta \left(\frac \pi 4 \sqrt{\frac Nr}\right)\]
    steps, measuring the resulting quantum state yields an element in $\cap_{i=0}^{k-1}\calA_i$ with probability
    \[\psucc(T) = \mathbb E\left[a_0^2 (T)\right] \ge 1 - \delta. \]
    Moreover, the expected total number of applications of $G_1, \dots, G_{k-1}$ is
    \[\mathcal O\left(\frac k\delta \cdot \frac{m}{r}\right).\]
\end{corollary}
\begin{proof}
    According to \Cref{the:2}, we have 
    \begin{equation}
    \begin{aligned}
        1 - \mathbb{E}[a^{2}_{0}(T)] 
        &\leq \frac{4r}{N} + \left(2 + \frac4p\right)\frac{m-r}{\sqrt{rN}} - \frac2p \sqrt{\frac{m-r}{N}}\\
        &= \frac{4r}{N} + \frac{2(m-r)}{\sqrt{rN}} + \frac{\frac{4(m-r)}{\sqrt{rN}} + 2\sqrt{\frac{m-r}{N}}}{p} = \delta.
    \end{aligned}
\end{equation}
Therefore,
\[\psucc(T) = \mathbb E\left[a_0^2 (T)\right] \ge 1 - \delta.\]
The expected total number of applications of $G_1, \dots, G_{k-1}$ is 
\[pT(k-1) = \mathcal O\left(\frac k\delta \cdot \frac mr\right).\]
\end{proof}
\subsection{Optimality of the random Grover search algorithm}\label{sec:optimality}
In this section, we analyze the optimality of our approach in terms of query complexity
\begin{theorem}[Optimality of the random Grover search algorithm]\label{optimal}
Consider $k\ge 1$ distinct Grover operators $G_{0},G_{1},\dots,G_{k-1}$. For any sequence $G_{i_1},G_{i_2},\dots,G_{i_t}$, suppose that the corresponding success probability satisfies
\[
\psucc(t)
=
\left|
\bra{\calI_0}
G_{i_t}G_{i_{t-1}}\cdots G_{i_1}
\ket{\psi}
\right|^2
\geq
1-\delta,
\]
where $\delta = \mathcal O(1/\sqrt{N})$. Then
\[
t
=
\Omega\left(\frac{\pi}{4}\sqrt{\frac{N}{r}}\right).
\]
\end{theorem}   
\begin{proof}
Let
\[
\ket{\psi_t}
:=
G_{i_t}G_{i_{t-1}}\cdots G_{i_1}\ket{\psi}
=
\sum_{j=0}^{2^k-1} a_j(t)\ket{\calI_j},
\]
and define
\[
\sin\theta_t
:=
|a_0(t)|
=
\left|\braket{\calI_0}{\psi_t}\right|,
\qquad \theta_t\in\left[0,\frac{\pi}{2}\right].
\]
In particular,
\[
\sin\theta_0
=
|a_0(0)|
=
\alpha_0.
\]
Since
\[
O_i\ket{\calI_j}
=
(-1)^{1 + j_i}\ket{\calI_j},
\]
where we use the binary representation of $j = (j_{k-1}\cdots j_1j_0)_2$, and
\[ G_{i_t} = (2\ket{\psi}\bra{\psi}-I)O_{i_t}, \qquad \ket{\psi} = \sum_{j=0}^{2^k-1}\alpha_j\ket{\calI_j},\]
we obtain
\[ \begin{aligned}
a_0(t)
&=
\bra{\calI_0}
(2\ket{\psi}\bra{\psi}-I)
O_{i_t}
\ket{\psi_{t-1}} = 2\braket{\calI_0}{\psi} \bra{\psi} O_{i_t} \ket{\psi_{t-1}} - \bra{\calI_0} O_{i_t} \ket{\psi_{t-1}}\\
&=
2\alpha_0
\sum_{j=0}^{2^k-1}
(-1)^{1 + j_{i_t}}
\alpha_j a_j(t-1)
-
(-1)^{1 + 0_{i_t}}a_0(t-1)\\
&= a_0(t-1) + 2\alpha_0 \sum_{j=0}^{2^k-1}
(-1)^{1 + j_{i_t}}
\alpha_j a_j(t-1),
\end{aligned}
\]
where we use the fact that $0_i = 0$ for all $i = 0, \dots, k-1$ in the last inequality. Taking absolute values gives
\[
\sin\theta_t
=
|a_0(t)|
=
\left|
a_0(t-1)
+
2\alpha_0
\sum_{j=0}^{2^k-1}
(-1)^{1 + j_{i_t}}
\alpha_j a_j(t-1)
\right|.
\]
Therefore,
\[\begin{aligned}
    \sin\theta_t &\leq (1-2\alpha_0^{2})|a_{0}(t-1)| + 2\alpha_0 \sum_{j=1}^{2^{k}-1}\alpha_{j}| a_{j}(t-1)| \\
    &\leq (1-2\alpha_0^{2})|a_{0}(t-1)| + 2\alpha_0 \sqrt{\sum_{j=1}^{2^{k}-1}\alpha_{j}^{2}}\sqrt{\sum_{j=1}^{2^{k}-1} a_{j}(t-1)^{2}} \\
    &\leq (1-2\alpha_0^{2})| a_{0}(t-1)| + 2\alpha_0 \sqrt{1-\alpha_0^{2}}\sqrt{1- a_{0}(t-1)^{2}} \\
    & = \cos(2\theta_0) \sin\theta_{t-1} + \sin(2\theta_0) \cos\theta_{t-1}\\
    & = \sin(\theta_{t-1} + 2\theta_{0}).
\end{aligned}\]
It can be shown by induction that
\begin{equation}\label{eq:theta_t}
    \theta_t\leq (2t+1)\theta_0 ,
\quad
\forall t\leq \frac{\pi}{4\theta_0}-\frac12.
\end{equation}
The case \(t=0\) is immediate. If the relation holds for \(t-1\), then for $t \le \frac{\pi}{4\theta_0}-\frac12$, we have
\[
\theta_{t-1}+2\theta_0
\leq
(2t-1)\theta_0+2\theta_0
=
(2t+1)\theta_0
\leq
\frac{\pi}{2}.
\]
By the monotonicity of the sine function on $\left(0, \frac\pi2\right]$ and $\sin\theta_t \leq \sin(\theta_{t-1}+2\theta_0)$, we have
\[
\theta_t
\leq
\theta_{t-1}+2\theta_0
\leq
(2t+1)\theta_0.
\]
Consequently, the success probability at time $t$ satisfies
\[\psucc(t) = \sin^2 \theta_t \le \sin^2 \left((2t+1)\theta_0\right), \quad \forall t\le \frac{\pi}{4\theta_0}-\frac12. \]
Given $\sin^2\left((2t+1)\theta_0\right) \ge \psucc(t) > 1-\delta$, we further obtain
\begin{equation}\label{eq:lower_bound}
    \begin{aligned}
        t &\geq \frac{\arcsin\sqrt{1-\delta}}{2\theta_0} - \frac{1}{2}
        \geq \frac{\frac{\pi}{2}-\sqrt{2\delta}}{2(\sin\theta_0 + \sin^{3}\theta_0)} - \frac{1}{2}
        \geq (\frac{\pi}{4} - \sqrt{\delta})\cdot \sqrt{\frac{N}{r}}\cdot\frac{N}{N+r} - \frac{1}{2},
    \end{aligned}
    \end{equation}
    where we use the inequalities $\arcsin\sqrt{1-\delta}\geq \pi/2-\sqrt{2\delta}$ and $\theta_0 \leq \sin\theta_0 + \sin^3\theta_0$.
    Let $\delta = \mathcal O(1/\sqrt{N})$, we have 
    \[t = \Omega\left(\frac\pi 4\sqrt{\frac Nr}\right).\] 
    
\end{proof}

\section{Numerical Experiments}\label{sec:numeric}
\FloatBarrier
In this section, we present several numerical experiments to verify the theoretical results. The code for all simulations is publicly available\footnote{\url{https://github.com/greenttt-cell/RGS_code}}. 

\subsection{Comparison between deterministic and random strategies}\label{sec:comp_deter_rand}
In this group of experiments, we compare a deterministic periodic sequence with the randomized strategy for two Grover operators. The deterministic sequence is
\[G_0, G_1, G_0, G_0, G_1, G_0, G_0, \cdots\]
so that \(G_0\) and \(G_1\) are applied with frequencies \(2/3\) and \(1/3\), respectively. 
In the randomized strategy, \(G_0\) and \(G_1\) are independently selected with probabilities \(2/3\) and \(1/3\). For each parameter setting, we perform $50$ independent randomized trials and plot the empirical mean evolution of the target success probability \(a_0^2(t)\).

\begin{figure}[t]
    \centering

    \begin{subfigure}{0.48\textwidth}
        \centering
        \includegraphics[width=\linewidth]{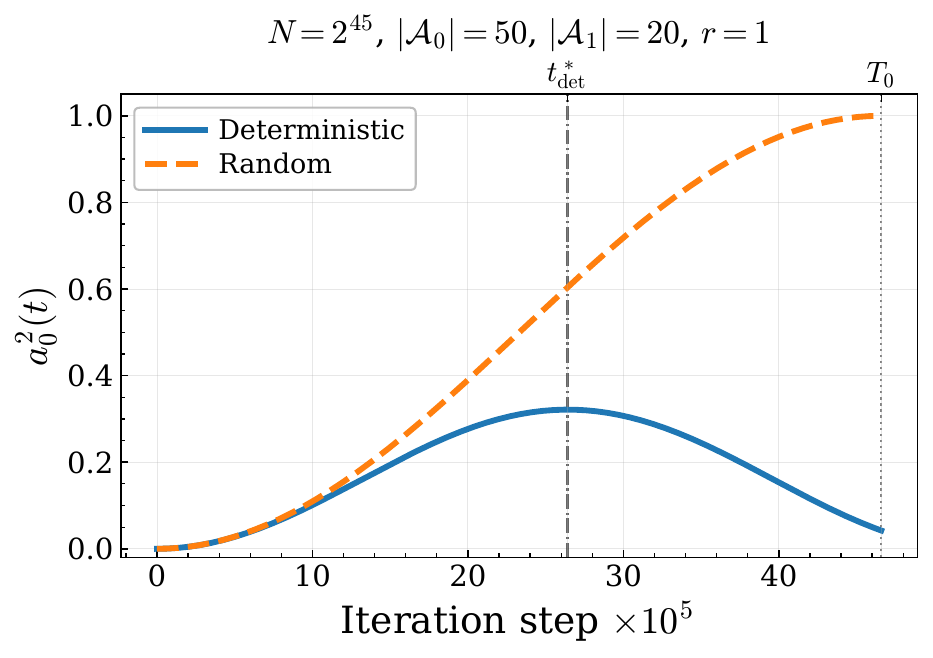}
        \caption{}
        \label{fig:cmp1}
    \end{subfigure}
    \hfill
    \begin{subfigure}{0.48\textwidth}
        \centering
        \includegraphics[width=\linewidth]{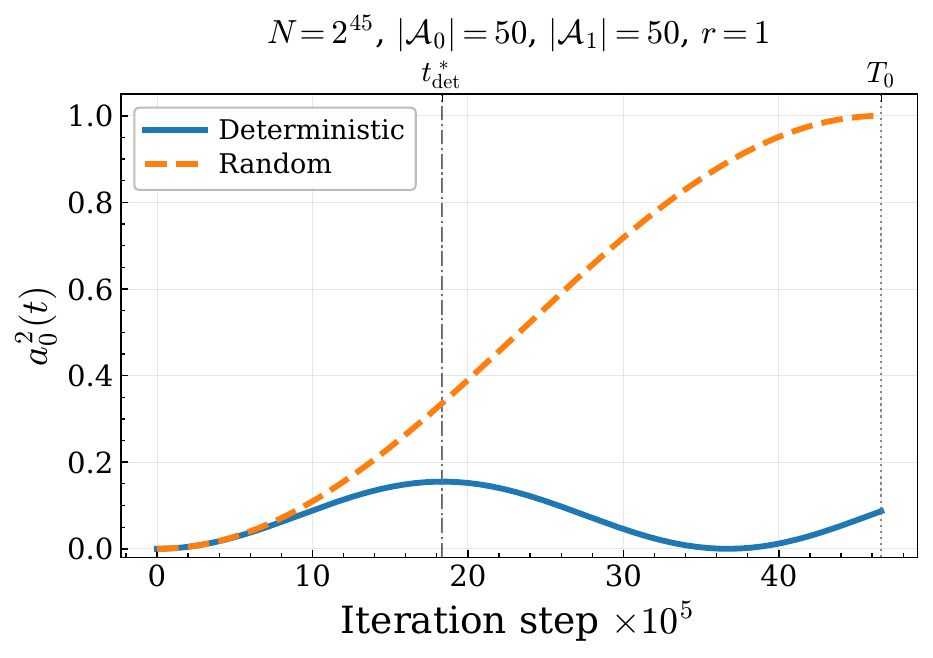}
        \caption{}
        \label{fig:cmp2}
    \end{subfigure}

    \vspace{0.6em}

    \begin{subfigure}{0.48\textwidth}
        \centering
        \includegraphics[width=\linewidth]{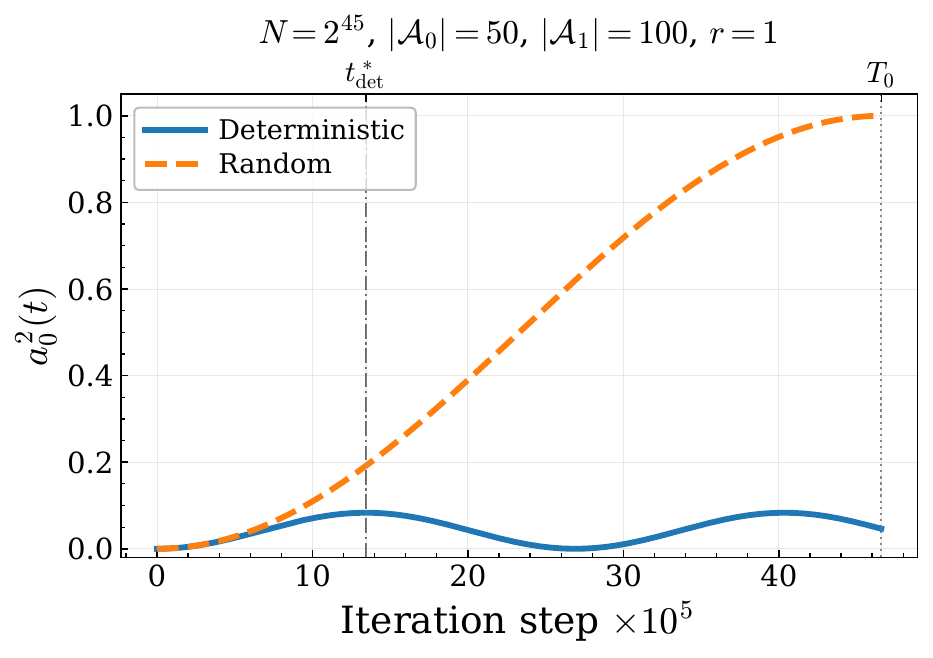}
        \caption{}
        \label{fig:cmp3}
    \end{subfigure}
    \hfill
    \begin{subfigure}{0.48\textwidth}
        \centering
        \includegraphics[width=\linewidth]{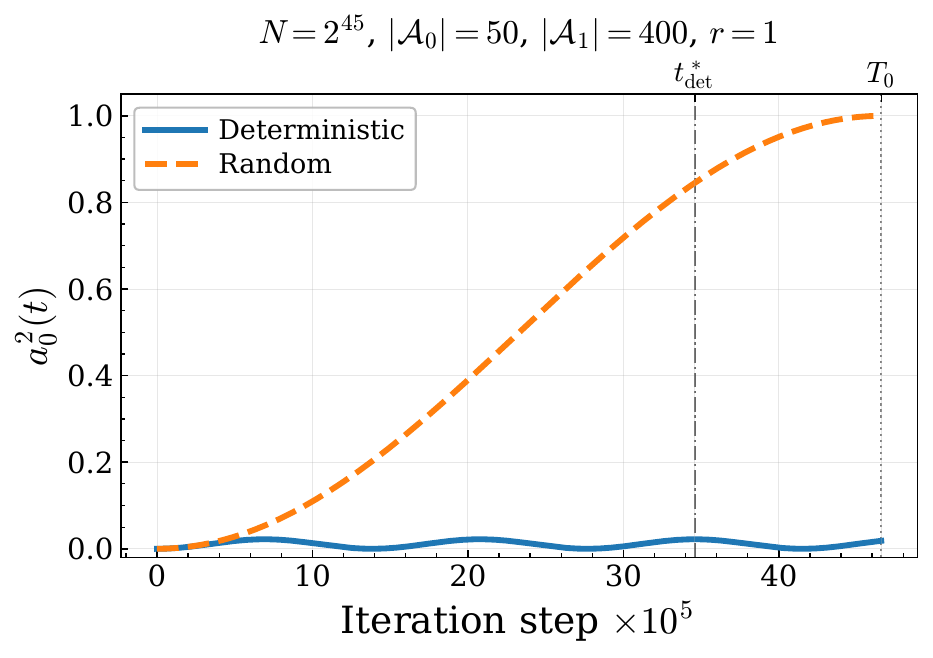}
        \caption{}
        \label{fig:cmp4}
    \end{subfigure}

    \caption{Comparison between the deterministic periodic sequence and the corresponding randomized strategy under four parameter settings. In each panel, the solid curve represents the deterministic sequence, while the dashed curve shows the empirical mean over 50 independent randomized trials. 
    The vertical dotted line marks the stopping time \(T\), and the dash-dotted vertical line marks the fist peak time \(t_{\mathrm{det}}^\ast\) of the deterministic curve. The corresponding values of $t_{\mathrm{det}}^\ast$ are $2.6\times 10^6$, $1.8\times 10^6$, $1.3\times 10^6$, and $0.7\times 10^6$ for the four panels, respectively.}
    \label{fig:comparison_all}
\end{figure}

\Cref{fig:comparison_all} compares the two approaches under four parameter settings. In all cases, the deterministic periodic sequence exhibits persistent oscillations and fails to amplify the success probability close to $1$. By contrast, the randomized strategy yields a nearly monotone increase in the success probability and reaches a value close to $1$ near the stopping time $T$. One may attempt to improve the deterministic approach by terminating the evolution at its first peak time \(t_{\mathrm{det}}^\ast\) and boosting the overall success probability through repeated measurements. However, this strategy remains less efficient than the randomized scheme. For example, in \Cref{fig:cmp2}
\[t_{\mathrm{det}}^\ast > 1.8\times 10^6, \quad a_0^2(t_{\mathrm{det}}^\ast) < 0.2,\quad \text{and}\quad  T < 5.0\times 10^6.\]
Repeating the deterministic procedure three times yields a total cost exceeding $5.4\times 10^6$ iterations, while the overall success probability remains below
\[1 - (1-0.2)^3 = 0.488.\]
Thus, even when combined with repetition, the deterministic strategy is unable to match the performance of the randomized approach.

The advantage of randomized strategy persists as \(|\mathcal A_1|\) varies over \(20,50,100,\) and \(400\) while the overlap size $|\calA_0\cap \calA_1|$ is fixed. As \(|\mathcal A_1|\) increases, the deterministic sequence becomes progressively less effective, with its maximum success probability decreasing substantially. By contrast, the randomized strategy consistently achieves a success probability close to one.

\subsection{Sampling Grover operators with a biased distribution}
We next validate \Cref{lemma: delta_k}, which shows that a suitably biased sampling distribution can maintain a high success probability while substantially reducing the use of expensive Grover operators. All simulations are implemented using the reduced matrix representations of the Grover operators in the basis \(\{\ket{\calI_j}:j = 0, \dots, 2^k-1\}\). In all experiments, the evolution starts from the uniform superposition state and is run for 
\[T=\left\lfloor\frac{\pi}{4}\sqrt{\frac{N}{r}}\right\rfloor\]
iterations. At each step, a Grover operator is sampled independently according to the distribution specified in \Cref{eq:cond_p1}. For each parameter setting, we perform $50$ independent trails and estimate the expected success probability
\[\psucc(t)=\mathbb{E}\left[a_0^2(t)\right], \quad t = 1, \dots, T.\]

We first consider the two-operator setting. Fixing \(N=2^{40}\), \(|\calA_1|=100\), and \(r=1\), we vary \(|\calA_0|\in \{100, 1000\}\) and the tolerance parameter \(\delta\in\{0.1,0.3,0.5\}\). The results are shown in \Cref{fig:delta_compare}. In all cases, the success probability increases steadily and reaches a value exceeding the theoretical guarantee \(1-\delta\) near the stopping time \(T\). As expected, smaller values of \(\delta\) lead to higher final success probabilities. For comparison, we also include the evolution obtained by repeatedly applying only \(G_0\). The resulting success probability remains substantially lower, indicating that even infrequent applications of \(G_1\) play an essential role in the amplification process.
\begin{figure}[htbp]
    \centering

    \begin{subfigure}{0.45\textwidth}
        \centering
        \includegraphics[width=\linewidth]{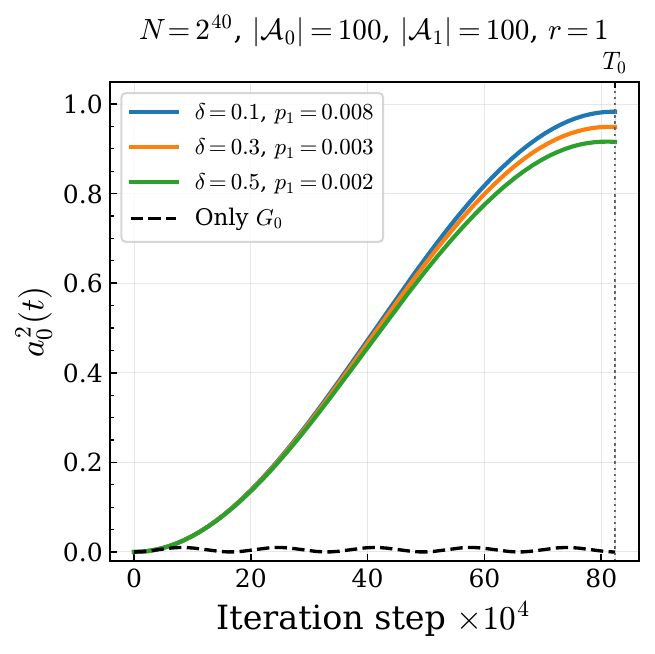}
        \caption{}
        \label{fig:delta_compare_a}
    \end{subfigure}
    \hfill
    \begin{subfigure}{0.45\textwidth}
        \centering
        \includegraphics[width=\linewidth]{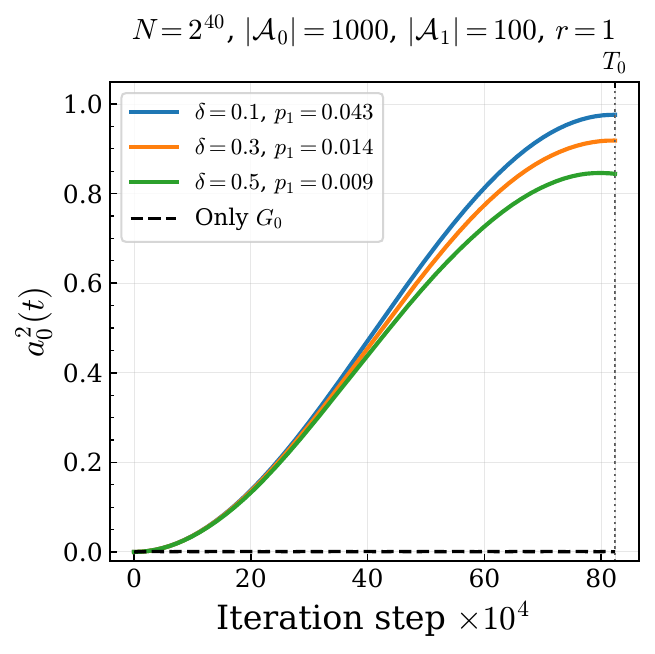}
        \caption{}
        \label{fig:delta_compare_b}
    \end{subfigure}

    \caption{
    Success-probability evolution for different values of \(\delta\), with fixed \(N=2^{40}\), \(|\calA_1|=100\), and \(r=1\). The solid curves show the empirical averages over 50 independent trials, while the dashed curve corresponds to repeatedly applying only \(G_0\). The vertical dotted line marks the stopping time \(T\).
    }
    \label{fig:delta_compare}
\end{figure}

We next consider a three-operator setting with \(N=2^{40}\), \(|\calA_0|=100\), \(|\calA_1|=200\), \(|\calA_2|=50\), and \(r=1\). The tolerance parameter $\delta$ is chosen from \(\{0.005,0.15,0.4\}\). As shown in \Cref{fig:res5a}, the success probability again approaches one near \(T\) and consistently exceeds the lower bound \(1-\delta\), confirming the prediction of \Cref{lemma: delta_k}.

To examine the sampling distribution, \Cref{fig:res5b} reports the empirical usage ratios of the three Grover operators. In all tested cases, the inexpensive operator \(G_0\) accounts for the overwhelming majority of iterations, while \(G_1\) and \(G_2\) are used only sparingly. Moreover, increasing \(\delta\) further reduces the usage frequency of the expensive operators while maintaining a high success probability. These observations demonstrate that the proposed biased sampling strategy successfully concentrates queries on low-cost oracles without sacrificing the amplification performance.
\begin{figure}[htbp]
    \centering

    \begin{subfigure}{0.45\textwidth}
        \centering
        \includegraphics[width=\linewidth]{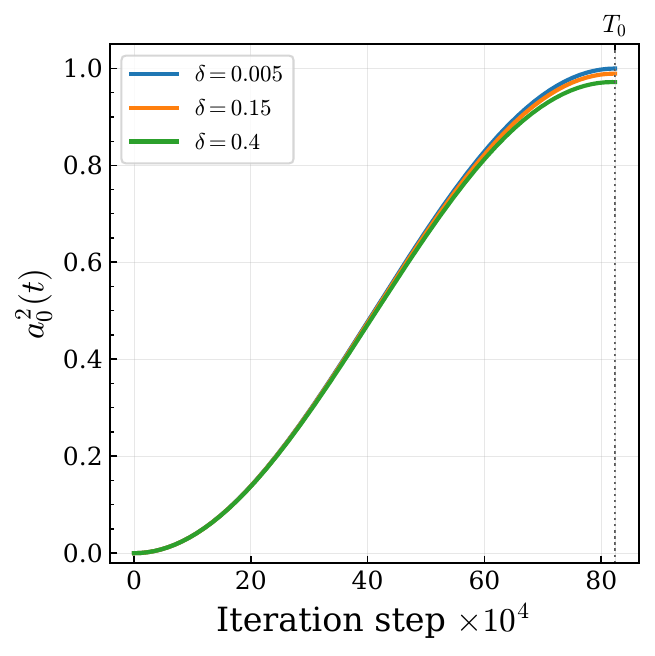}
        \caption{}
        \label{fig:res5a}
    \end{subfigure}
    \hfill
    \begin{subfigure}{0.45\textwidth}
        \centering
        \includegraphics[width=\linewidth]{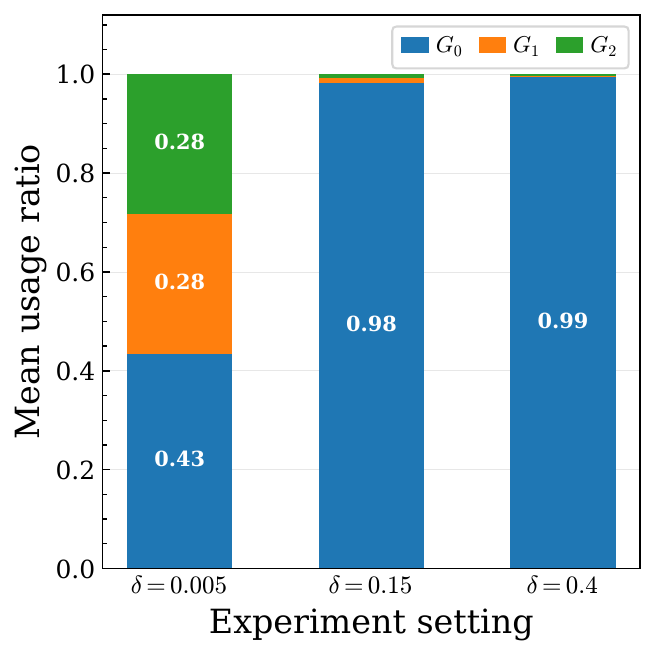}
        \caption{}
        \label{fig:res5b}
    \end{subfigure}

    \caption{
    Numerical results for
    \(N=2^{40}\), \(|\calA_0|=100\), \(|\calA_1|=200\), \(|\calA_2|=50\), \(r=1\), and
    \(\delta\in\{0.005,0.15,0.4\}\).
    (a) Empirical mean success probability for different values of \(\delta\).
    (b) Corresponding mean usage ratios of \(G_0\), \(G_1\), and \(G_2\).}
    \label{fig:res5}
\end{figure}

\section{Conclusion}\label{sec:conclusion}
In this work, we studied a randomized variant of Grover search for solving intersection problems defined by multiple constraints. Instead of constructing a global oracle, which can be costly, our approach directly exploits the Grover operators associated with the individual constraints. An each step, a Grover operator is selected according to a prescribed sampling distribution and applied to the quantum state. 

Under mild assumptions, we proved that the expected success probability approaches one after
\[\Theta\left(\frac\pi 4\sqrt{\frac Nr}\right)\]
iterations by leveraging the special structure of the expected Grover operator. Furthermore, we showed that a highly biased sampling strategy is sufficient to achieve near-optimal performance, making the proposed framework particular attractive when the implementation costs of different Grover operators vary significantly. We also established that the randomized algorithm is asymptotically optimal in query complexity: any deterministic sequence of Grover operators that achieves a comparable success probability must incur the same asymptotic complexity. Numerical simulations were conducted to validate the theoretical results and illustrate the effectiveness of the proposed approach across different parameter regimes. 

These findings indicate that randomization provides an effective and analyzable mechanism for combining multiple Grover operators, offering a practical alternative when the construction of a global oracle is prohibitively expensive. An interesting direction for future work is to explore the random framework on estimating the number of solutions. 
\appendix

    
    
\section{Supplementary Calculations}\label{sec:supply}
\subsection{Derivation of the expression of $\widehat a_0(t)$}\label{app:hat_a_0_t}
To compute $\widehat a_0(t)$ explicitly, we diagonalize the matrix $M$. 
The characteristic polynomial of $M$ is given by
\[
\lambda^2 - 2(\alpha_3^2 - \alpha_0^2)\lambda + 1 = 0,
\]
whose roots are $\lambda_{\pm} = (\alpha_3^2 - \alpha_0^2) \pm \Delta$, where
\[
\Delta = \sqrt{(1 - \alpha_0^2 - \alpha_3^2)^2 - 4\alpha_0^2 \alpha_3^2}.
\]
Thus, $M$ can be diagonalized as
\[
M = P 
\begin{bmatrix}
\lambda_+ & 0 \\
0 & \lambda_-
\end{bmatrix}
P^{-1},\]
where 
\[ P =
\begin{bmatrix}
2\alpha_0\alpha_3 & 2\alpha_0\alpha_3\\
\alpha_0^2+\alpha_3^2-1+\Delta &
\alpha_0^2+\alpha_3^2-1-\Delta
\end{bmatrix},\]
and 
\[P^{-1} = \frac{1}{4\alpha_0\alpha_3\Delta}
\begin{bmatrix}
-\bigl(\alpha_0^2+\alpha_3^2-1-\Delta\bigr) & 2\alpha_0\alpha_3\\
\alpha_0^2+\alpha_3^2-1+\Delta & -2\alpha_0\alpha_3
\end{bmatrix}.\]
    
Using the fact that  
\[
M^t = P 
\begin{bmatrix}
\lambda_+^t & 0 \\
0 & \lambda_-^t
\end{bmatrix}
P^{-1},
\]
we get 
\begin{equation}\label{eq:hat_a_0_t}
    \widehat a_0(t)
    = \frac{\alpha_0(1-\alpha_0^{2}+\alpha_3^2 + \Delta)(\alpha_3^2 - \alpha_0^2 + \Delta)^t 
    + \alpha_0(\alpha_0^{2}-\alpha_3^2 -1 + \Delta)(\alpha_3^2 - \alpha_0^2 - \Delta)^t}{2\Delta}.
\end{equation}

\subsection{Derivation of the expression of $W^n$}\label{app:w}
Let $c=\braket{u}{v}$. The matrix $W$ defined in \cref{eq:W} is of the form
\[
W=
\begin{bmatrix}
1 & 2c\\
-2\overline{c} & 1-4|c|^2
\end{bmatrix},
\]
and its characteristic polynomial is
\[
\lambda^2-(2-4|c|^2)\lambda+1=0.
\]
Define $\theta$ by $\cos\theta=1-2|c|^2$ (so that $\sin^2(\theta/2)=|c|^2$), yielding eigenvalues $\lambda_\pm=e^{\pm i\theta}$. By the Cayley--Hamilton theorem,
\[
W^2-2\cos\theta\,W+I=0,
\]
which implies the recurrence
\[
W^n=2\cos\theta\,W^{n-1}-W^{n-2},\quad n\ge2,\quad W^0=I,\ W^1=W.
\]
Solving it gives
\[
W^n=\frac{\sin(n\theta)}{\sin\theta}W-\frac{\sin((n-1)\theta)}{\sin\theta}I.
\]
Substituting $W$ yields
\[
W^n=
\frac{1}{\sin\theta}
\begin{bmatrix}
\sin(n\theta)-\sin((n-1)\theta) & 2c\,\sin(n\theta)\\
-2\overline{c}\,\sin(n\theta) & (1-4|c|^2)\sin(n\theta)-\sin((n-1)\theta)
\end{bmatrix}.
\]
Therefore,
\[
W^n=
\frac{1}{\sin\theta}
\begin{bmatrix}
\sin(n\theta)-\sin((n-1)\theta)
&
2\braket{u}{v}\sin(n\theta)
\\[1ex]
-2\braket{v}{u}\sin(n\theta)
&
\bigl(1-4|\braket{u}{v}|^2\bigr)\sin(n\theta)-\sin((n-1)\theta)
\end{bmatrix}
\]

\subsection{Additional estimates for deterministic prefix variations}\label{app:deterministic-prefix}
In main text, we consider the quantum state
\[
(G_0G_0G_1)^nG_0\ket{\psi},\quad \forall n = 0, 1, \dots,
\]
and show that the success probability $\psucc^{(n)} := \left|\bra{\calI_0}(G_0G_0G_1)^nG_0\ket{\psi}\right|^2$ can be significantly smaller than $1$ for all $n$. Here we consider the following two variant sequences
\[
G_1(G_0G_0G_1)^nG_0\ket{\psi}\quad  \text{and}\quad G_0G_1(G_0G_0G_1)^nG_0\ket{\psi}.
\]
First, since $G_1^\dagger
=
O_1(2\ket{\psi}\bra{\psi}-I)$, we have
\[
\begin{aligned}
\|G_1^\dagger \ket{\calI_0} - \ket{\calI_0}\|_2 = \left\| O_1 \left(2\sqrt{\frac{r}{N}}\ket{\psi} - \ket{\calI_0}\right)-\ket{\calI_0}\right\|_2 =
2\sqrt{\frac{r}{N}}\|O_1\ket{\psi}\|_2 \leq 2\sqrt{\frac{r}{N}}.
\end{aligned}
\]
Similarly,
\[
\|G_0^\dagger \ket{\calI_0} - \ket{\calI_0}\|_2
\leq
2\sqrt{\frac{r}{N}}.
\]
Consequently,
\[
\begin{aligned}
\|G_1^\dagger G_0^\dagger \ket{\calI_0} - \ket{\calI_0}\|_2
&=
\|(G_0(G_1-I)+(G_0-I))^\dagger\ket{\calI_0}\|_2 \\
&\leq
\|G_0\|_2
\cdot
\|(G_1-I)^\dagger\ket{\calI_0}\|_2
+
\|(G_0-I)^\dagger\ket{\calI_0}\|_2\leq
4\sqrt{\frac{r}{N}}.
\end{aligned}
\]
For the first prefix variation, we have
\[
\begin{aligned}
& \left|
\bra{\calI_0}
G_1(G_0G_0G_1)^nG_0
\ket{\psi}
\right|^2\\
\leq &
\left(
\left|
\bra{\calI_0}
(G_1-I)(G_0G_0G_1)^nG_0
\ket{\psi}
\right|
+
\left|
\bra{\calI_0}
(G_0G_0G_1)^nG_0
\ket{\psi}
\right|
\right)^2 \\
\leq & 
\left(
\|(G_1-I)^\dagger\ket{\calI_0}\|_2
\cdot
\|(G_0G_0G_1)^nG_0\ket{\psi}\|_2
+
\sqrt{\psucc^{(n)}}
\right)^2 \\
\leq & 
\left(
2\sqrt{\frac{r}{N}}
+
\sqrt{\psucc^{(n)}}
\right)^2.
\end{aligned}
\]
For the second prefix variation, we similarly have
\[
\begin{aligned}
& \left|
\bra{\calI_0}
G_0G_1(G_0G_0G_1)^nG_0
\ket{\psi}
\right|^2\\
\leq & 
\left(
\left|
\bra{\calI_0}
(G_0G_1-I)(G_0G_0G_1)^nG_0
\ket{\psi}
\right|
+
\left|
\bra{\calI_0}
(G_0G_0G_1)^nG_0
\ket{\psi}
\right|
\right)^2 \\
\leq & 
\left(
\|(G_0G_1-I)^\dagger\ket{\calI_0}\|_2
\cdot
\|(G_0G_0G_1)^nG_0\ket{\psi}\|_2
+
\sqrt{\psucc}
\right)^2 \\
\leq &
\left(
4\sqrt{\frac{r}{N}}
+
\sqrt{\psucc}
\right)^2.
\end{aligned}
\]
Therefore, both variations are controlled by the base-sequence success probability, up to an additional error of order $\mathcal O(\sqrt{r/N})$. Together with the upper bound of $\psucc^{(n)}$, it shows that the two variations also fail to amplify the success probability close to $1$ in the regime $r \ll r_1 < N$.


\bibliographystyle{siamplain}
\bibliography{references}
\end{document}